\documentclass{article}
\usepackage{ijcai16}

\usepackage{times}

\usepackage{amsmath, amssymb, amsthm}
\usepackage{pgf}
\usepackage{pgfplots}
\usepackage{pgfkeys}
\usepackage{subfig}
\usepackage{url}

\usetikzlibrary{patterns}

\title{Combining the $k$-CNF and XOR Phase-Transitions	\thanks{The author list has been sorted alphabetically by last name; this should not be used to determine the extent of authors' contributions.}}
\author{Jeffrey M. Dudek\\Rice University \And Kuldeep S. Meel\\Rice University \And  Moshe Y. Vardi\\Rice University
}

\begin{document}

\maketitle

\setcounter{topnumber}{2}
\setcounter{bottomnumber}{2}
\setcounter{totalnumber}{2}

\pgfkeys{
	/cnfxor/.is family, /cnfxor,
	default/.style = {k = {k}, n = {n}, r = {r}, s = {s}, cnfNum = {}, xorNum = {}},
	k/.estore in = \clauseK,
	n/.estore in = \clauseN,
	r/.estore in = \clauseR,
	s/.estore in = \clauseS,
	cnfNum/.estore in = \clauseCNFOverride,
	xorNum/.estore in = \clauseXOROverride,
}

\newcommand\F[1][]{
	\pgfkeys{/cnfxor, default, #1}
	F_{\clauseK}(\clauseN, 
		\ifx\clauseCNFOverride\empty\relax
			\clauseR \clauseN
		\else
			\clauseCNFOverride
		\fi
	)
}

\newcommand\Q[1][]{
	\pgfkeys{/cnfxor, default, #1}
	Q(\clauseN, 
		\ifx\clauseXOROverride\empty\relax
			\clauseS \clauseN
		\else
			\clauseXOROverride
		\fi
	)
}

\newcommand\FQ[1][]{
	\pgfkeys{/cnfxor, default, #1}
	\psi_{\clauseK}(\clauseN, 
		\ifx\clauseCNFOverride\empty\relax
			\clauseR \clauseN
		\else
			\clauseCNFOverride
		\fi,
		\ifx\clauseXOROverride\empty\relax
			\clauseS \clauseN
		\else
			\clauseXOROverride
		\fi
	)
}

\renewcommand{\P}[0]{\ensuremath{\mathsf{Pr}}}
\newcommand{\Var}[1]{\ensuremath{\mathsf{Var}\left[#1\right]}}
\newcommand{\Covar}[2]{\ensuremath{\mathsf{Cov}\left[#1, #2\right]}}
\newcommand{\E}[1]{\ensuremath{\mathsf{E}\left[#1\right]}}

\renewcommand{\L}[0]{\Lambda}
\newcommand{\Lfull}[0]{\Lambda_b(k, r)}
\newcommand{\Lfullest}[0]{\Lambda_b(1/2, k, r)}

\newcommand{\SMT}{\ensuremath{\mathsf{SMT}}}
\newcommand{\SAT}{\ensuremath{\mathsf{SAT}}}

\newcommand{\seq}[2]{\{#1_#2\}_{#2 = 1}^{\infty}}
\newcommand{\ceil}[1]{\lceil #1 \rceil}

\newcommand{\CryptoMiniSAT}{\ensuremath{\mathsf{CryptoMiniSAT}}}

\newtheorem{theorem}{Theorem}
\newtheorem{lemma}[theorem]{Lemma}
\newtheorem{corollary}[theorem]{Corollary}

\newtheorem{conjecture}{Conjecture}
\newtheorem{speculation}{Speculation}
\everymath{\textstyle}

\begin{abstract}
The runtime performance of modern \SAT~solvers on random $k$-CNF formulas is deeply connected with the `phase-transition' phenomenon seen empirically in the satisfiability of random $k$-CNF formulas. Recent universal hashing-based approaches to sampling and counting crucially depend on the runtime performance of \SAT~solvers on
formulas expressed as the conjunction of both $k$-CNF and XOR constraints (known as $k$-CNF-XOR formulas), but the behavior of random $k$-CNF-XOR formulas is unexplored in prior work. In this paper, we present the first study of the satisfiability of random $k$-CNF-XOR formulas. We show empirical evidence of a surprising phase-transition that follows a linear trade-off between $k$-CNF and XOR constraints. Furthermore, we prove that a phase-transition for $k$-CNF-XOR formulas exists for $k=2$ and (when the number of $k$-CNF constraints is small) for $k > 2$.
\end{abstract}

\section{Introduction}\label{sec:introduction}

The Constraint-Satisfaction Problem (CSP) is one of the most fundamental problem in 
computer science, with a wide range of applications arising from diverse areas 
such as artificial intelligence, programming languages, biology and the like~\cite{Apt03,Dec03}.
The problem is, in general, NP-complete, and the study of run-time behavior of CSP
techniques is a topic of major interest in AI, cf.~\cite{DM94}. Of specific interest
is the behavior of CSP solvers on random problems~\cite{CKT91}. Specifically, a deep
connection was discovered between the density (ratio of clauses to variables) of random
propositional CNF fixed-width (fixed number of literals per clause) formulas and the runtime 
behavior of SAT solvers on such formulas~\cite{MSL92,CA93,KS94}. The key experimental findings
are: (1) as the density of random CNF instances increases, the probability of satisfiability 
decreases with a precipitous drop, believed to be a phase-transition, around the point 
where the satisfiability probability is 0.5, and (2) instances at the phase-transition point 
are particularly challenging for DPLL-based SAT solvers.  Indeed, phase-transition instances 
serve as a source of challenging benchmark problems in {\SAT} competitions~\cite{BDHJ14}. 
The connection between runtime performance and the satisfiability phase-transition has propelled 
the study of such phase-transition phenomena over the past two decades~\cite{Achlioptas09}, 
including detailed studies of how SAT solvers scale at different densities~\cite{CDSSV03,MH15}.

For random $k$-CNF formulas, where every clause contains exactly $k$ literals, experiments
suggest a specific phase-transition density, for example, density 4.26 for random 3-SAT,
but establishing this analytically has been highly challenging~\cite{CP13a}, and it has been
established only for for $k=2$~\cite{CR92,Goerdt96} and all large enough $k$~\cite{DSS15}. 
A phase-transition phenomenon has also been identified in random XOR formulas (conjunctions of 
XOR constraints). 
Creignou and Daud{\'e}~\shortcite{CD99} proved a phase-transition at density 1 for variable-width 
random XOR formulas.  Creignou and Daud{\'e}~\shortcite{CD03} also proved the existence of a phase 
transition for random $\ell$-XOR formulas (where each XOR-clause contains exactly $\ell$ literals), for $\ell\geq 1$,  
without specifying an exact location for the phase-transition. Dubois and Mandler \shortcite{DM02} independently identified the location of a phase transition for random 3-XOR formulas. More recently, Pittel and 
Sorkin~\shortcite{PS15} identified the location of the phase-transition for $\ell$-XOR formulas for $\ell > 3$.

Despite the abundance of prior work on the phase-transition phenomenon in the satisfiability of 
random $k$-CNF formulas and XOR formulas, no prior work considers the satisfiability of 
random formulas with \emph{both} $k$-clauses and variable-width XOR-clauses together, 
henceforth referred as $k$-CNF-XOR formulas. Recently, successful hashing-based 
approaches to the fundamental problems of constrained sampling and counting employ {\SAT} 
solvers to solve $k$-CNF-XOR formulas~\cite{CMV13a,CMV13b,CMV14,CMV16,MVCFSFIM16}. Unlike previous approaches 
to sampling and counting, hashing-based approaches provide strong theoretical guarantees and 
scale to real-world instances involving formulas with {\em hundreds of thousands} of variables. 
The scalability of these approaches crucially depends on the runtime performance of 
{\SAT} solvers in handling $k$-CNF-XOR formulas ~\cite{IMMV15}. Moreover, since the phase-transition behavior of $k$-CNF constraints have been analyzed to explain runtime behavior of {\SAT} solvers~\cite{AchCoj08}, we believe that analysis of the phase-transition phenomenon for $k$-CNF-XOR formula is the first step towards demystifying the runtime behavior of CNF-XOR solvers such as CryptoMiniSAT~\cite{SNC09} and thus explain the runtime behavior of hashing-based algorithms.

The primary contribution of this work is the first study of phase-transition phenomenon in the satisfiability of random $k$-CNF-XOR formulas, henceforth referred to as the $k$-CNF-XOR phase-transition. In particular:
\begin{enumerate}
	\item We present (in Section~\ref{sec:experiments}) experimental evidence for a $k$-CNF-XOR phase-transition that follows a linear trade-off between $k$-CNF clauses and XOR clauses.
	\item We prove (in Section~\ref{sec:theoretical}) that the $k$-CNF-XOR phase-transition exists when the ratio of $k$-CNF clauses to variables is small. This fully characterizes the phase-transition when $k=2$.
	\item We prove (in Section~\ref{sec:theoretical}) upper and lower bounds on the location of the $k$-CNF-XOR  phase-transition region.
	\item We conjecture (in Section~\ref{sec:conjecture}) that the exact location of a phase-transition for $k \geq 3$ follows the linear trade-off between $k$-CNF and XOR clauses seen experimentally.
\end{enumerate}

\section{Notations and Preliminaries} \label{sec:prelims}
Let $X = \{X_1, \cdots, X_n\}$ be a set of propositional variables and 
let $F$ be a formula defined over $X$. A \emph{satisfying assignment} or 
\emph{witness} of $F$ is an assignment of truth values to the variables in $X
$ such that $F$ evaluates to true. Let $\#F$ denote the number of satisfying 
assignments of $F$. We say that $F$ is \emph{satisfiable} (or \emph{sat.}) if $\#F > 0$ and 
that $F$ is \emph{unsatisfiable} (or \emph{unsat.}) if $\#F = 0$.

We use $\P(X)$ to denote the probability of event $X$. 
We say that an infinite sequence of random events $E_1, E_2, \cdots $ occurs 
\emph{with high probability} (denoted, w.h.p.) if $\lim\limits_{n \to \infty} \P(E_n) = 1$. We use $\E{Y}$ and $\Var{Y}$ to denote respectively the expected value and variance of a random variable $Y$. We use $\Covar{Y}{Z}$ to denote the covariance of random variables $Y$ and $Z$. We use $o_k(1)$ to denote a term which converges to 0 as $k \to \infty$.

A $k$-{\emph clause} is the disjunction of $k$ literals out of $\{ X_1, \cdots, X_n\}$, 
with each variable possibly negated. For fixed positive integers $k$ and $n$ and 
a nonnegative real number $r$, let the random variable $\F$ 
denote the formula consisting of the conjunction of $\ceil{rn}$ $k$-clauses, with each clause chosen 
uniformly and independently from all $\binom{n}{k}2^k$ possible $k$-clauses over $n$ variables.

The early experiments on $\F$~\cite{MSL92,CA93,KS94} led to the following conjecture:
\begin{conjecture}[Satisfiability Phase-Transition Conjecture]
	For every integer $k \geq 2$, there is a critical ratio $r_k$ such that:
	\begin{enumerate}
		\item If $r < r_k$, then $\F$ is satisfiable w.h.p.
		\item If $r > r_k$, then $\F$ is unsatisfiable w.h.p.
	\end{enumerate}
\end{conjecture}
The Conjecture was quickly proved for $k=2$, where $r_2 = 1$ \cite{CR92,Goerdt96}. 
In recent work, Ding, Sly, and Sun established the Satisfiability Phase Transition Conjecture 
for all sufficiently large $k$~\cite{DSS15}. The Conjecture has remained elusive for small values 
of $k \geq 3$, although values for these critical ratios $r_k$ can be estimated experimentally, 
e.g., $r_3$ seems to be near $4.26$.

An XOR-clause over $n$ variables is the `exclusive or' of either 0 or 1 together with a subset 
of the variables $X_1$, $\cdots$, $X_n$. An XOR-clause including 0 (respectively, 1) evaluates to true if and only 
if an odd (respectively, even) number of the included variables evaluate to true. 
Note that all $k$-clauses contain \emph{exactly} $k$ variables, whereas the number of variables 
in an XOR-clause is not fixed; a uniformly chosen XOR-clause over $n$ variables contains 
$\frac{n}{2}$ variables in expectation. 

For a fixed positive integer $n$ and a nonnegative 
real number $s$, let the random variable $\Q$ denote 
the formula consisting of the conjunction of $\ceil{sn}$ XOR-clauses, with each clause chosen uniformly 
and independently from all $2^{n+1}$ XOR-clauses over $n$ variables. 
Creignou and Daude~\shortcite{CD99,CD03} proved a phase-transition in the satisfiability of $\Q$: if $s < 1$ then $\Q$ is satisfiable w.h.p., 
while if $s > 1$ then $\Q$ is unsatisfiable w.h.p. 

A $k$-CNF-XOR formula is the conjunction of some number of $k$-clauses and XOR-clauses. 
For fixed positive integers $k$ and $n$ and fixed nonnegative real numbers $r$ and 
$s$, let the random variable $\FQ$ denote the 
formula consisting of the conjunction of $\ceil{rn}$ $k$-clauses and $\ceil{sn}$ XOR-clauses, 
with each clause chosen uniformly and independently from all possible $k$-clauses and 
XOR-clauses over $n$ variables. (The motivation for using
fixed-width clauses and variable-width XOR-clauses comes from the hashing-based approaches to
constrained sampling and counting discussed in Section \ref{sec:introduction}.) Although random $k$-CNF formulas and XOR formulas have been well studied separately, 
no prior work considers the satisfiability of random mixed formulas arising 
from conjunctions of $k$-clauses and XOR-clauses. 

\section{Experimental Results} \label{sec:experiments}
To explore empirically the behavior of the satisfiability of $k$-CNF-XOR 
formulas, we built a prototype implementation in Python that employs the 
{\CryptoMiniSAT}\footnote{\url{http://www.msoos.org/cryptominisat4/}}~\cite{SNC09} 
solver to check satisfiability of $k$-CNF-XOR formulas. We chose {\CryptoMiniSAT} 
due to its ability to handle the combination of $k$-clauses and XOR-clauses efficiently~\cite{CMV14,CFMSV14}. 
The objective of the experimental setup was to empirically determine the behavior of 
$\P(\FQ \text{ is sat})$ with respect to $r$ and $s$, the $k$-clause and XOR-clause 
densities respectively, for fixed $k$ and $n$. 

\subsection{Experimental Setup}
We ran 11 experiments with various values of $k$ and $n$. For $k=2$, we ran experiments for $n \in \{25,50,100,150\}$. For $k=3$, we ran experiments for $n \in \{25, 50, 100\}$. For $k = 4$ and $k=5$, we ran experiments for $n \in \{25, 50\}$. We were not able to run experiments for values of $n$ significantly larger than those listed above: at some $k$-clause and XOR-clause densities, the run-time of {\CryptoMiniSAT} scaled far beyond our computational capabilities.

In each experiment, the XOR-clause 
density $s$ ranged from 0 to 1.2 in increments of 0.02. Since the location of phase-transition 
for $k-$CNF depends on $k$, the range of $k$-clause density $r$ also depends 
on~$k$. For $k=3$, $r$ ranged from from 0 to 6 in increments of 0.04;
for $k=5$, $r$ ranged from 0 to 26 in increments of 0.43, and the like.

To uniformly choose a $k$-clause we uniformly selected without replacement 
$k$~out of the variables $\{X_1, \cdots, X_n\}$. For each selected variable $X_i$, 
we include exactly one of the literals $X_i$ or $\neg X_i$ in the $k$-clause, 
each with probability $1/2$. The disjunction of these $k$ literals is 
a uniformly chosen $k$-clause.  To uniformly choose an XOR-clause, we include 
each variable of $\{X_1, \cdots, X_n\}$ with probability $1/2$ in a set 
$A$ of variables. Additionally we include in $A$ exactly one of $0$ or $1$, 
each with probability $\frac{1}{2}$. The `exclusive-or' of all elements of $A$ 
is a uniformly chosen XOR-clause.
\begin{figure}
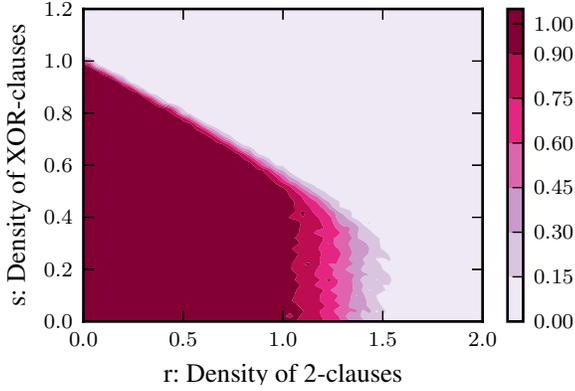

	\centering
\begingroup
\makeatletter

\makeatother
\endgroup

	\caption{Phase transition for $2$-CNF-XOR formulas  \label{fig:2cnfxor}}
\end{figure}
For each assignment of values to $k$, $n$, $r$, and $s$, we evaluated satisfiability, 
using {\CryptoMiniSAT}, of $100$ uniformly generated formulas of $\FQ$ by constructing 
the conjunction of $\ceil{rn}$ $k$-clauses and $\ceil{sn}$ XOR-clauses, with each clause chosen uniformly 
and independently as described above. The percentage of satisfiable formulas gives us an empirical estimate of $\P(\FQ~\text{is satisfiable})$.

Each experiment was run on a node within a high-performance computer cluster. 
These nodes contain 12-processor cores at 2.83 GHz each with 48 GB of RAM per node. 
Each formula was given a timeout of 1000 seconds.

\subsection{Results}

We present scatter plots demonstrating the behavior of satisfiability of $k$-CNF-XOR formulas. 
For lack of space, we present results only for three experiments\footnote{
	The data from all experiments is available at \url{http://www.cs.rice.
	edu/CS/Verification/Projects/CUSP/}
}.
The plots for $k=2,3$ and $5$
are shown in Figure~\ref{fig:2cnfxor},~\ref{fig:3cnfxor}, and ~\ref{fig:5cnfxor}
 respectively. The value of $n$ is set to 150, 100, and 50 respectively for the three experiments above. 

Each figure is a 2D plot, representing the observed probability that $\FQ$ is 
satisfiable as the density of $k$-clauses $r$ and the density of XOR-clauses $s$ varies. 
The x-axis indicates the density of $k$-clauses $r$. The y-axis indicates the density of XOR-clauses $s$. 
The dark (respectively, light) regions represent clause densities where almost all (respectively, no) sampled formulas were satisfiable.

Note that $\FQ$ consists only of XOR clauses when $r=0$. Examining the figures along the line $r=0$ the phase-transition location is 
around ($r=0$, $s=1$), which matches previous theoretical results on the phase-transition for XOR formulas~\cite{CD99}. 
Likewise, $\FQ[xorNum=0] = \F$ and, by examining the figures along the line $s=0$, we observe phase-transition 
locations that match previous studies on the phase-transition 
for $k$-CNF formulas for $k=2,3,~\text{and}~5$ ~\cite{Achlioptas09}. Note that the phase-transition we observe for $2$-CNF formulas is slightly above the true location at $s=1$ \cite{CR92,Goerdt96}; the correct phase-transition point for $2$-CNF formulas is observed only when the number of variables is above 4096 \cite{Wilson2000}.

\begin{figure}
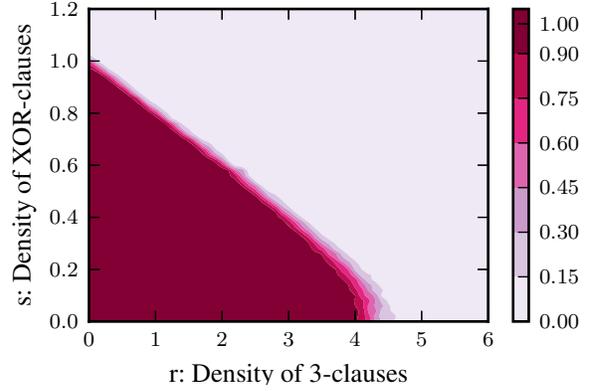

	\centering
\begingroup
\makeatletter

\makeatother
\endgroup

	\caption{Phase transition for $3$-CNF-XOR formulas \label{fig:3cnfxor}}
\end{figure}
\begin{figure}
	\centering
\begingroup
\makeatletter
%
\makeatother
\endgroup

	\caption{Phase transition for $5$-CNF-XOR formulas  \label{fig:5cnfxor}}
\end{figure}
In all the plots, we observe a large triangular region where the probability that $\FQ$ is 
satisfiable is nearly 1. We likewise observe a 
separate region where the observed probability that $\FQ$ is satisfiable is nearly 0. 
More surprisingly, the shared boundary between the two regions for large areas of the 
plots seems to be a constant-slope line. A closer examination of this line at the bottom-right corners of the figures for $k=2$ and $k=3$, where the $k$-clause density is large, reveals that the line appears to ``kink'' and abruptly change slope. We discuss this further in Section~\ref{sec:conjecture}.

\section{Establishing a Phase-Transition} \label{sec:theoretical}
The experimental results presented in Section~\ref{sec:experiments} empirically demonstrate the existence of a  $k$-CNF-XOR phase-transition. Theorem~\ref{thm:mainresult:tight} shows that the $k$-CNF-XOR phase-transition exists when the density of $k$-clauses is small. In particular, the function $\phi_k(r)$ (defined in Lemma \ref{lemma:abbe}) gives the location of a phase-transition between a region of satisfiability and a region of unsatisfiability in random $k$-CNF-XOR formulas.

\renewcommand{\theenumi}{(\alph{enumi})}
\begin{theorem}
	\label{thm:mainresult:tight}
	Let $k \geq 2$.	There is a function $\phi_k(r)$, a constant $\alpha_k \geq 1$, and a countable set of real numbers $\mathcal{C}_k$ (all defined in Lemma \ref{lemma:abbe}) such that for all $r \in [0, \alpha_k) \backslash \mathcal{C}_k$ and $s \geq 0$:
	\begin{enumerate}
		\item \label{thm:mainresult:tight:sat} If $s < \phi_k(r)$, then w.h.p. $\FQ$ is satisfiable.
		\item \label{thm:mainresult:tight:unsat} If $s > \phi_k(r)$, then w.h.p. $\FQ$ is unsatisfiable.
	\end{enumerate}
\end{theorem}
\begin{proof}
	Part \ref{thm:mainresult:tight:sat} follows directly from Lemma \ref{lemma:mainresult:tight:sat}. Part \ref{thm:mainresult:tight:unsat} follows directly from Lemma \ref{lemma:mainresult:tight:unsat}. The proofs of these lemmas are presented in Sections \ref{sec:theoretical:lower} and \ref{sec:theoretical:upper} respectively.
\end{proof}

$\phi_k(r)$ is the \emph{free-entropy density} of $k$-CNF, drawing on concepts from spin-glass theory \cite{Gogioso14}. From the expression for $\phi_k(r)$ in Lemma \ref{lemma:abbe}, it is easily verified that $\phi_k(0) = 1$ and that $\phi_k(r)$ is a monotonically decreasing function of $r$. Thus when the $k$-clause density ($r$) is 0, Theorem \ref{thm:mainresult:tight} says that an XOR-clause density of $1$ is a phase-transition for XOR-formulas, matching previously known results \cite{CD99}. As the $k$-clause density increases, $\phi(r)$ is decreasing and so the XOR-clause density required to reach the phase-transition decreases.

Theorem \ref{thm:mainresult:tight} fully characterizes the random satisfiability of $\FQ$ when $r < 1$. In the case $k=2$, prior results on random $2$-CNF satisfiability characterize the rest of the region. If $r > 1$, then $\F[k=2]$ is unsatisfiable w.h.p. \cite{CR92,Goerdt96} and so the $2$-clauses within $\FQ[k=2]$ are unsatisfiable w.h.p. without considering the XOR-clauses. Therefore $\FQ[k=2]$ is unsatisfiable w.h.p. if $r > 1$. This, together with Theorem \ref{thm:mainresult:tight}, proves that $\phi_2(r)$ is the complete location of the 2-CNF-XOR phase-transition.

Moreover, Lemma \ref{lemma:alpha} shows that $\alpha_k \geq (1 - o_k(1)) \cdot 2^k \ln(k) / k$ (where $o_k(1)$ denotes a term that converges to 0 as $k \to \infty$) and so Theorem \ref{thm:mainresult:tight} shows that a phase-transition exists until near $r = 2^k \ln(k) / k$ for sufficiently large $k$.

For small $k \geq 3$, the region $r < 1$ characterized by Theorem~\ref{thm:mainresult:tight} is only a small portion of the region where the subset of $k$-clauses remains satisfiable. Moreover, the location of the phase-transition $\phi_k(r)$ given by Theorem \ref{thm:mainresult:tight} is difficult to compute directly. Theorem~\ref{thm:mainresult:loose} gives explicit lower and upper bounds on the location of a phase-transition region.

\renewcommand{\theenumi}{(\alph{enumi})}
\begin{theorem} \label{thm:mainresult:loose}
	Let $k \geq 3$. There is a function $\Lfull$ (defined in Lemma \ref{lemma:achlioptas_numsol})
	such that for all $s \geq 0$ and $r \geq 0$:	
	\begin{enumerate}
		\item \label{thm:mainresult:loose:sat} If $s < \frac{1}{2} \log_2(\Lfull)$ and $r < 2^k \ln(2) - \frac{1}{2}((k+1)\ln(2) + 3)$, then w.h.p. $\FQ$ is satisfiable.
		\item \label{thm:mainresult:loose:unsat} If $s > r \log_2(1 - 2^{-k}) + 1$, then w.h.p. $\FQ$ is unsatisfiable.
	\end{enumerate}
\end{theorem}
\begin{proof}
	Part \ref{thm:mainresult:loose:sat} follows directly from Lemma \ref{lemma:mainresult:loose:sat}. Part \ref{thm:mainresult:loose:unsat} follows directly from Lemma \ref{lemma:mainresult:loose:unsat}. The proofs of these lemmas are presented in Sections \ref{sec:theoretical:lower} and \ref{sec:theoretical:upper} respectively.
\end{proof}

Both the upper bound $r \log_2(1 - 2^{-k}) + 1$ and (using the expression for $\Lfull$ in Lemma \ref{lemma:achlioptas_numsol}) the lower bound $\frac{1}{2} \log_2(\Lfull)$ are linear in $r$. 
When the $k$-clause density $r$ is 0, Theorem \ref{thm:mainresult:loose} agrees with Theorem \ref{thm:mainresult:tight}. As the $k$-clause density increases past $\Theta(2^k)$, Theorem \ref{thm:mainresult:loose} no longer gives a lower bound on the location of a possible phase-transition.

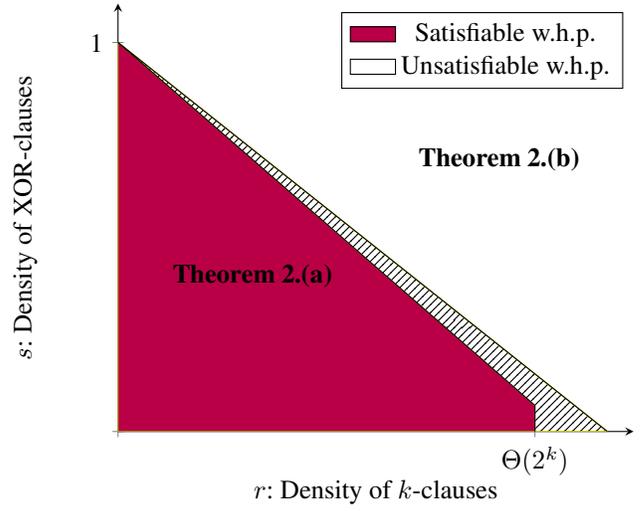
\begin{figure}[t]
	\centering
	\begin{tikzpicture}
	\label{key}
	\begin{axis} [
	axis lines = left,
	xmin=0, xmax=23, 
	xtick={0},
	extra x ticks={18.60},
	extra x tick style={xticklabel={$\Theta(2^k)$}},
	ymin=0, ymax=1.1,
	ytick={0},
	extra y ticks={1},
	extra y tick style={yticklabel={1}},
	xticklabels={,,},
	yticklabels={,,},
	xlabel={$r$: Density of $k$-clauses},
	ylabel={$s$: Density of XOR-clauses}
	]
	\addlegendimage{area legend, fill={rgb:red,8;blue,3}}
	\addlegendentry{Satisfiable w.h.p.}
	\addlegendimage{area legend, fill=white}
	\addlegendentry{Unsatisfiable w.h.p.}
	
	\addplot[patch,patch type=triangle, color=white]
	coordinates {(0, 0)(0, 1)(21.83,0)};
	\addplot[patch,patch type=triangle, pattern=north east lines]
	coordinates {(0, 0)(0, 1)(21.83,0)};
	
	\addplot[patch,patch type=rectangle, color={rgb:red,8;blue,3}]
	coordinates {(0,0)(0,1)(18.6,0.06694)(18.6,0)};
	
	\addplot[
	color=black,
	domain=0:21.83
	]
	{1 - 0.0458 * x};
	\addplot[color=black,] coordinates {(18.60, 0)(18.6,0.06694)};
	\addplot[color=black,] coordinates {(0, 1)(18.6,0.06694)};
	
	\node[anchor=center] at (axis cs:6, 0.4) {\textbf{Theorem \ref{thm:mainresult:loose}.\ref{thm:mainresult:loose:sat}}};
	\node[anchor=center] at (axis cs:17, 0.7) {\textbf{Theorem \ref{thm:mainresult:loose}.\ref{thm:mainresult:loose:unsat}}};
	\end{axis}
	\end{tikzpicture}
	\caption{Satisfiability of $\FQ$ as $n \rightarrow \infty$ \label{fig:thmresults}}
\end{figure}

\subsection{A Proof of the Lower Bound}
\label{sec:theoretical:lower}
We now establish Theorem \ref{thm:mainresult:tight}.\ref{thm:mainresult:tight:sat} and Theorem \ref{thm:mainresult:loose}.\ref{thm:mainresult:loose:sat}, which follow directly from Lemma~\ref{lemma:mainresult:tight:sat} and Lemma~\ref{lemma:mainresult:loose:sat} respectively. 

The key idea in the proof of these lemmas is to decompose $\FQ$ into independently generated $k$-CNF and XOR formulas, so that $\FQ = \F \land \Q$. We can then bound the number of solutions to $\F$ from below with high probability and bound from below the probability that $\F$ becomes unsatisfiable after including XOR-clauses on top of $\F$. 

The following three lemmas achieve the first of the two tasks. The first, Lemma \ref{lemma:abbe}, gives a tight bound on $\# \F$ for small $k$-clause densities.
\begin{lemma}
	\label{lemma:abbe}
	Let $k \geq 2$ and let $\alpha_k$ be the supremum of $\{r : \exists \delta > 0~\text{s.t.}~\P(\F~\text{is unsat.}) \leq O(1/(\log n)^{1+\delta})\}$. Then $\alpha_k \geq 1$. Furthermore, there exists a countable set of real numbers $\mathcal{C}_k$ such that for all $r \in [0, \alpha_k) \backslash \mathcal{C}_k$:
	\begin{enumerate}
		\item \label{lemma:abbe:converge} The sequence $\frac{1}{n} \E{\log_2(\#\F)~|~ \F~\text{is sat.}}$ converges to a limit as $n \to \infty$. Let $\phi_k(r)$ be this limit.
		\item \label{lemma:abbe:lower} For all $\epsilon > 0$, w.h.p. $(2^{\phi_k(r)-\epsilon})^n \leq \# \F$.
		\item \label{lemma:abbe:upper} For all $\epsilon > 0$, w.h.p. $(2^{\phi_k(r)+\epsilon})^n \geq \# \F$.
	\end{enumerate}
\end{lemma}
\begin{proof}
	These proofs are given in \cite{AM14}. $\alpha_k \geq 1$ is given as \textbf{Remark 2}. Part (a) is given as \textbf{Theorem 3}. Parts (b) and (c) are given as \textbf{Theorem 1}.
\end{proof}

We abuse notation to let $\phi_k(r)$ denote the limit of the sequence in Lemma \ref{lemma:abbe}.\ref{lemma:abbe:converge} for all $r > 0$, although a~priori this sequence may not converge for $r \geq \alpha_k$. Later work refined the value of $\alpha_k$ in Lemma \ref{lemma:abbe} for sufficiently large $k$ and so extended the tight bound on $\# \F$. In particular, Lemma \ref{lemma:alpha} implies that $\alpha_k \geq (1 - o_k(1)) \cdot 2^k \ln(k) / k$.
\begin{lemma}
	\label{lemma:alpha}
	Let $k \geq 2$. For all $r \geq 0$, if $r \leq (1 - o_k(1)) \cdot 2^k \ln(k) / k$ then $\P(\F~\text{is sat.}) \geq 1 - O(1/n)$.
\end{lemma}
\begin{proof}
	The proof of this is given as \textbf{Theorem 1.3} of \cite{CR13}.
\end{proof}

It is difficult to compute $\phi_k(r)$ directly. Instead, Lemma~\ref{lemma:achlioptas_numsol} provides a weaker but explicit lower bound on $\# \F$.
\begin{lemma}
	\label{lemma:achlioptas_numsol}
	Let $k \geq 3$, $\epsilon > 0$, and $r \geq 0$. Let $\beta_k$ be the smallest positive solution to $\beta_k (2-\beta_k)^{k-1} = 1$ and define $\Lfull = 4 (((1 - \beta_k /2)^k - 2^{-k})^2 / (1-\beta_k)^{k})^r$.
	
	If $r < 2^k \ln(2) - \frac{1}{2}((k+1)\ln(2) + 3)$,
	then w.h.p.\ $\frac{1}{2}(\Lfull - \epsilon)^{n/2} \leq \# \F$.
\end{lemma}
\begin{proof}
	The proof of this is given on page 264 of \cite{ACR11} within Section 6 (\textbf{Proof of Theorem 6}); the definition of $\Lfull$ is given as equation (20).
\end{proof}

Before we analyze how the solution space of $\F$ interacts with the solution space of $\Q$, we must characterize the solution space of $\Q$. The following lemma shows that the solutions of $\Q$ are pairwise independent, meaning that a single satisfying assignment of $\Q$ gives no information on other satisfying assignments.
\begin{lemma}
	\label{lemma:gomes_xor:pairwise}
	Let $n \geq 0$ and $s \geq 0$. If $\sigma$ and $\sigma'$ are distinct assignments of truth values to the variables $\{X_1, \cdots, X_n\}$:
	\begin{enumerate}
		\item $\P(\sigma~\text{satisfies}~\Q) = 2^{-\ceil{sn}}$
		\item $\P(\sigma~\text{satisfies}~\Q~|~ \sigma'~\text{satisfies}~\Q) = 2^{-\ceil{sn}}$
	\end{enumerate}
\end{lemma}
\begin{proof}
	The proof of this is given in the proof of Lemma 1 of \cite{GSS07}.
\end{proof}

The following lemma bounds from below the probability that a formula $H$ (in Lemma \ref{lemma:cnfxor_sat} we take $H = \F$) remains satisfiable after including XOR-clauses on top of $H$. This result and proof is similar to Corollary 3 from ~\cite{GSS06}.
\begin{lemma}
	\label{lemma:gomes_xor:unsatbound}
	Let $\alpha \geq 1$, $s \geq 0$, $n \geq 0$, and let $H$ be a formula defined over $\{X_1, \cdots, X_n\}$.
	Then
	$\P(H \land \Q \text{ is satisfiable} ~|~ \# H \geq 2^{\ceil{sn} + \alpha}) \geq 1 - 2^{-\alpha}$.
\end{lemma}
\begin{proof}
	Let $R$ be the set of all truth assignments to the variables in $X$ that satisfy $H$; there are $\#H$ such truth assignments. For every truth assignment $\sigma \in R$, let $Y_\sigma$ be a 0-1 random variable that is 1 if $\sigma$ satisfies $H \land \Q$ and 0 otherwise. Note that $\Var{Y_\sigma} = \E{Y_\sigma^2} - \E{Y_\sigma}^2 \leq \E{Y_\sigma^2}$. Since $Y_\sigma$ is a 0-1 random variable, $Y_\sigma^2 = Y_\sigma$ and thus $\Var{Y_\sigma} \leq \E{Y_\sigma}$.
	
	Let $\sigma$ and $\sigma'$ be distinct truth assignments in $R$. By Lemma \ref{lemma:gomes_xor:pairwise}, 	$\E{Y_\sigma Y_{\sigma'}} = \E{Y_\sigma} \E{Y_{\sigma'}} = 2^{-\ceil{sn}} \cdot 2^{-\ceil{sn}}$. Thus $\Covar{Y_\sigma}{Y_{\sigma'}} = \E{Y_\sigma Y_{\sigma'}} - \E{Y_\sigma} \E{Y_{\sigma'}} = 0$.
	
	Let the random variable $Y$ be the number of solutions to $H \land \Q$, so $Y = \# (H \land \Q) = \sum_{\sigma} Y_\sigma$. Thus $\Var{Y} = \Var{\sum_{\sigma} Y_\sigma} = \sum_{\sigma} \Var{Y_\sigma} + \sum_{\sigma \neq \sigma'} \Covar{Y_\sigma}{Y_{\sigma'}}$. Since the covariance of $Y_\sigma$ and $Y_{\sigma'}$ is 0 for all pairs of distinct truth assignments $\sigma$ and $\sigma'$ in $R$, we get that $\Var{Y} = \sum_{\sigma} \Var{Y_\sigma}$. Since $\Var{Y_\sigma} \leq \E{Y_\sigma}$ for all truth assignments $\sigma$ in $R$, we get that $\Var{Y} \leq \sum_{\sigma} \E{Y_\sigma}$. Since $\E{Y} = \E{\sum_{\sigma} Y_\sigma} = \sum_{\sigma} \E{Y_\sigma}$, we conclude that $\Var{Y} \leq \E{Y}$. Moreover, since $\E{Y_\sigma} = \P(\sigma~\text{satisfies}~\Q) = 2^{-\ceil{sn}}$ we get $\E{Y} = \# H \cdot 2^{-\ceil{sn}}$.
	
	Let the event $\neg E_n$ denote that $H \land \Q$ is unsatisfiable. Thus if $\neg E_n$ occurs then then $Y=0$ and so $\left|Y - \E{Y}\right| \geq \E{Y}$. This implies that $\P(\neg E_n) \leq \P(\left|Y - \E{Y}\right| \geq \E{Y})$. Chebyshev's inequality says that $\P(\left|Y- \E{Y}\right| \geq \E{Y}) \leq \Var{Y}~/~\E{Y}^2$. It follows that $\P(\neg E_n) \leq \Var{Y}~/~\E{Y}^2$. Since $\Var{Y} \leq \E{Y}$, we get that $\P(\neg E_n) \leq \E{Y}^{-1}$. Therefore by plugging in the value for $\E{Y}$ we get $\P(\neg E_n) \leq (\# H)^{-1} \cdot 2^{\ceil{sn}}$. 
	
	Finally, if we assume that $\#H \geq 2^{\ceil{sn} + \alpha}$ then $(\# H)^{-1} \cdot 2^{\ceil{sn}} \leq 2^{-\alpha}$. Therefore $\P(\neg E_n ~|~ \#H \geq 2^{\ceil{sn} + \alpha}) \leq 2^{-\alpha}$.
\end{proof}

Using the key behavior of XOR-clauses described in Lemma \ref{lemma:gomes_xor:unsatbound}, we can transform lower bounds (w.h.p.) on the number of solutions to $\F$ into lower bounds on the location of a possible $k$-CNF-XOR phase-transition.
\begin{lemma}
	\label{lemma:cnfxor_sat}
	Let $k \geq 2$, $s \geq 0$, and $r \geq 0$. Let $B_1$, $B_2$, $\cdots$ be an infinite convergent sequence of positive real numbers such that $B_i^n \leq \#\F$ occurs w.h.p.\ for all $i \geq 1$. If $s < \log_2(\lim_{i \to \infty} B_i)$, then w.h.p.\ $\FQ$ is satisfiable.
\end{lemma}
\begin{proof}
	For all integers $n \geq 0$, let the event $E_n$ denote the event when $\FQ$ is satisfiable. We would like to show that $\P(E_n)$ converges to $1$ as $n \to \infty$. 
	
	The general idea of the proof follows. We first decompose $\FQ$ as $\FQ = \F \land \Q$. Let the event $L_n$ denote the event when the number of solutions of $\F$ is bounded from below (by a lower bound to be specified later). We show that $L_n$ occurs w.h.p.. Next, we use Lemma \ref{lemma:gomes_xor:unsatbound} to bound from below the probability that $\F \land \Q$ remains satisfiable given that $\F$ has enough solutions; we use this to show that $\P(E_n~|~L_n)$ converges to $1$ as $n \to \infty$. Finally, we combine these results to prove that $\P(E_n)$ converges to $1$.
	
	Since $2^s < \lim_{i \to \infty} B_i$, there is some integer $i \geq 1$ such that $2^s < B_i$. Define the event $L_n$ as the event when $\#\F \geq B_i^n$. Then $L_n$ occurs w.h.p.\ by hypothesis.
	
	Next, we show that $\P(E_n~|~L_n)$ converges to 1. Choose $\delta > 0$ and $N > 0$ such that $2^{s + \delta + 1/N} < B_i$; we can always find sufficiently small $\delta$ and sufficiently large $N$ such that this holds. Since we are concerned only with the behavior of $\P(E_n ~|~ L_n)$ in the limit, we can restrict our attention only to large enough $n$. In particular, consider $n > 2N$. Then we get that $2^{sn + \delta n + 2} < B_i^n$ and so $2^{\ceil{sn} + \delta n + 1} < B_i^n$. Let $\alpha = \delta n + 1$, so that $2^{\ceil{sn} + \alpha} \leq B_i^n$. Then Lemma \ref{lemma:gomes_xor:unsatbound} says that $\P(E_n ~|~ L_n) \geq 1 - 2^{-\delta n -1}$. Since $1 - 2^{-\delta n - 1}$ converges to $1$ as $n \to \infty$, $\P(E_n ~|~ L_n)$ must also converge to $1$.
	
	Thus both $\P(E_n ~|~ L_n)$ and $\P(L_n)$ converge to 1 as $n \to \infty$. Since $\P(E_n \cap L_n) = \P(E_n~|~L_n) \cdot \P(L_n)$, this implies that $\P(E_n \cap L_n)$ also converges to 1. Since $\P(E_n \cap L_n) \leq \P(E_n) \leq 1$, this implies that $\P(E_n)$ converges to 1.
\end{proof}

Finally, it remains only to use Lemma \ref{lemma:cnfxor_sat} to obtain bounds on the $k$-CNF-XOR phase-transition. The tight lower bound on $\# \F$ from Lemma \ref{lemma:abbe}.\ref{lemma:abbe:lower} corresponds to a tight lower bound on the location of the phase-transition.
\begin{lemma}
	\label{lemma:mainresult:tight:sat}
	Let $k \geq 2$, and let $\alpha_k$, $\mathcal{C}_k$, and $\phi_k(r)$ be as defined in Lemma \ref{lemma:abbe}. 
	For all $r \in [0, \alpha_k) \backslash \mathcal{C}_k$ and $s \in [0, \phi_k(r))$, $\FQ$ is satisfiable w.h.p..
\end{lemma}
\begin{proof}
	Let $B_i = 2^{\phi_k(r)-1/i}$. By Lemma \ref{lemma:abbe}.\ref{lemma:abbe:lower}, $B_i^n \leq \#\F$ w.h.p.\ for all $i \geq 1$. Furthermore, $\lim_{i \to \infty} B_i = 2^{\phi_k(r)}$ and so $s < \log_2(\lim_{i \to \infty} B_i)$. Thus $\FQ$ is satisfiable w.h.p.\ by Lemma \ref{lemma:cnfxor_sat}.
\end{proof}

The weaker lower bound on $\# \F$ from Lemma \ref{lemma:achlioptas_numsol} corresponds to a weaker lower bound on the location of the phase-transition.
\begin{lemma}
	\label{lemma:mainresult:loose:sat}
	Let $k \geq 3$, $s \geq 0$, and $r \geq 0$. If $r < 2^k \ln(2) - \frac{1}{2}(k+1)\ln(2) + \frac{3}{2}$ and $s < \frac{1}{2} \log_2(\Lfull)$, then $\FQ$ is satisfiable w.h.p..
\end{lemma}
\begin{proof}
	Let $B_i = (\Lfull - 1/i)^{1/2}$. This is an increasing sequence in $i$, so $\log_2(B_{i+1} / B_i)$ is positive for all $i \geq 1$. Consider one such $i \geq 1$ and define $N_i = 1 / \log_2(B_{i+1} / B_i)$. Then for all $n > N_i$ it follows that $2^{1/n} < B_{i+1} / B_i$ and so $B_i^n < \frac{1}{2} B_{i+1}^n$. By Lemma \ref{lemma:achlioptas_numsol}, $\frac{1}{2}B_{i+1}^n \leq \#\F$ w.h.p.\ and therefore $B_i^n < \frac{1}{2} B_{i+1}^n \leq \#\F$ w.h.p.\ as well.
		
	 Furthermore, $\lim_{i \to \infty} B_i = \Lfull^{1/2}$ and so $s < \log_2(\lim_{i \to \infty} B_i)$. Thus $\FQ$ is satisfiable w.h.p.\ by Lemma \ref{lemma:cnfxor_sat}.
\end{proof}

\subsection{A Proof of the Upper Bound}
\label{sec:theoretical:upper}

We now establish Theorem \ref{thm:mainresult:tight}.\ref{thm:mainresult:tight:unsat} and Theorem \ref{thm:mainresult:loose}.\ref{thm:mainresult:loose:unsat}, which follow directly from Lemma~\ref{lemma:mainresult:tight:unsat} and Lemma~\ref{lemma:mainresult:loose:unsat} respectively. 

Similar to Section \ref{sec:theoretical:lower}, the key idea in the proof of these lemmas is to decompose $\FQ$ into independently generated $k$-CNF and XOR formulas, so that $\FQ = \F \land \Q$. We can then bound the number of solutions to $\F$ from above with high probability and bound from below the probability that $\F$ becomes unsatisfiable after including XOR-clauses on top of $\F$.

The first of these two tasks is accomplished through Lemma \ref{lemma:abbe}.\ref{lemma:abbe:upper}, which gives a tight upper bound on $\#\F$ for small $k$-clause densities, and by Lemma \ref{lemma:numsol_bound}, which gives a weaker explicit upper bound on $\#\F$.
\begin{lemma}
	\label{lemma:numsol_bound}
	For all $\epsilon > 1$, $k \geq 2$, and $r \geq 0$, w.h.p.\ $\#\F < (2 \epsilon \cdot (1 - 2^{-k})^r)^n$.
\end{lemma}
\begin{proof}
	Let $X = \#\F$.	For a random assignment on $n$ variables $\sigma$, note that $\P(\sigma \text{ satisfies } \F[cnfNum=1]) = (1 - 2^{-k})$. Since the $\ceil{rn}$ $k$-clauses of $\F$ were chosen independently, this implies that $\E{X} = 2^n (1 - 2^{-k})^{\ceil{rn}}$. 
	
	By Markov's inequality, we get $\P(X \geq \epsilon^n \E{X}) \leq \E{X} / (\epsilon^n \E{X}) = \epsilon^{-n}$.
	Since $1 - 2^{-k} < 1$ and so $2^n \epsilon^n (1 - 2^{-k})^{\ceil{rn}} \leq 2^n \epsilon^n (1 - 2^{-k})^{rn}$, it follows that $\P(X \geq \epsilon^n 2^n (1 - 2^{-k})^{rn}) \leq \P(X \geq \epsilon^n \E{X}) \leq \epsilon^{-n}$. Thus $\displaystyle \lim_{n \to\infty} \P(X < \epsilon^n 2^n (1 - 2^{-k})^{rn}) = 1$.
\end{proof}

The following lemma bounds from below the probability that a formula $H$ (in Lemma \ref{lemma:cnfxor_unsat} we take $H = \F$) remains satisfiable after including XOR-clauses on top of $H$. This result and proof is similar to Corollary 1 from \cite{GSS06}.
\begin{lemma}
	\label{lemma:gomes_xor:satbound}
	Let $\alpha \geq 1$, $s \geq 0$, $n \geq 0$, and let $H$ be a formula defined over $X = \{X_1, \cdots, X_n\}$. Then 
	$\P(H \land \Q~\text{is unsatisfiable} ~|~ \#H \leq 2^{\ceil{sn} - \alpha}) \geq 1 - 2^{-\alpha}$.
\end{lemma}
\begin{proof}
	Let the random variable $Y$ denote $\# (H \land \Q)$ as in Lemma \ref{lemma:gomes_xor:unsatbound}. Markov's inequality implies that $\P(Y \geq 1) \leq \E{Y}$. Recall from Lemma \ref{lemma:gomes_xor:unsatbound} that $\E{Y} =\# H \cdot 2^{-\ceil{sn}}$, so $\P(Y \geq 1) \leq \# H \cdot 2^{-\ceil{sn}}$. If $\#H \leq 2^{\ceil{sn} - \alpha}$, then $\# H \cdot 2^{-\ceil{sn}} \leq 2^{-\alpha}$. Thus $\P(Y \geq 1 ~|~ \#H \leq 2^{	\ceil{sn} - \alpha}) \leq 2^{-\alpha}$. Since $H \land \Q$ is unsatisfiable exactly when $Y = 0$, we conclude  $\P(H \land \Q~\text{is unsatisfiable}) \geq 1 - 2^{-\alpha}$.
\end{proof}

Using the key behavior of XOR-clauses described in Lemma \ref{lemma:gomes_xor:satbound}, we can transform upper bounds (w.h.p.) on the number of solutions to $\F$ into upper bounds on the location of a possible $k$-CNF-XOR phase-transition.
\begin{lemma}
\label{lemma:cnfxor_unsat}
Let $k \geq 2$, $s \geq 0$, and $r \geq 0$. Let $B_1$, $B_2$, $\cdots$ be an infinite convergent sequence of positive real numbers such that $\#\F \leq B_i^n$ occurs w.h.p.\ for all $i \geq 1$. If $s > \log_2(\lim_{i \to \infty} B_i)$, then w.h.p.\ $\FQ$ is unsatisfiable.
\end{lemma}
\begin{proof}
	For all integers $n \geq 0$, let the event $\neg E_n$ denote the event when $\FQ$ is unsatisfiable. We would like to show that $\P(\neg E_n)$ converges to $1$ as $n \to \infty$. 
	
	The general idea of the proof follows. Note that $\FQ = \F \land \Q$ as in Lemma \ref{lemma:cnfxor_sat}. Let the event $U_n$ denote the event when the number of solutions of $\F$ is bounded from above (by an upper bound to be specified later). We show that $U_n$ occurs w.h.p.. Next, we use Lemma \ref{lemma:gomes_xor:satbound} to bound from below the probability that $\F\land\Q$ becomes unsatisfiable given that $\F$ has few solutions; we use this to show that $\P(\neg E_n ~|~ U_n)$ converges to $1$ as $n \to \infty$. Finally, we combine these results to prove that $\P(\neg E_n)$ converges to 1.
			
	Since $2^s > \lim_{i \to \infty} B_i$, there is some integer $i \geq 1$ such that $2^s > B_i$. Define the event $U_n$ as the event when $\#\F \leq B_i^n$. Then $U_n$ occurs w.h.p.\ by hypothesis.
	
	Next, we show that $\P(\neg E_n~|~U_n)$ converges to 1. Choose $\delta > 0$ and $N > 0$ such that $2^{s - \delta - 1/N} > B_i$. As in Lemma \ref{lemma:cnfxor_sat} we are concerned only with the behavior of $\P(\neg E_n~|~U_n)$ in the limit so we can restrict our attention only to large enough $n$. In particular, consider $n > N$. Then we get that $2^{\ceil{sn} - \delta n - 1} > 2^{sn - \delta n - n/N} > B_i^n$. Let $\alpha = \delta n + 1$, so that $ 2^{\ceil{sn} - \alpha} \geq B_i^n$. Then Lemma \ref{lemma:gomes_xor:satbound} says that $\P(\neg E_n ~|~ U_n) \geq 1 - 2^{-\delta n - 1}$. Since $1 - 2^{-\delta n - 1}$ converges to $1$ as $n \to \infty$, $\P(\neg E_n ~|~ U_n)$ must also converge to $1$. 
	
	Thus both $\P(\neg E_n~|~U_n)$ and $\P(U_n)$ converge to $1$ as $n \to \infty$. Since $\P(\neg E_n \cap U_n) = \P(\neg E_n ~|~ U_n) \cdot \P(U_n)$, this implies that $\P(\neg E_n \cap U_n)$ also converges to $1$. Since $\P(\neg E_n \cap U_n) \leq \P(\neg E_n) \leq 1$, this implies that $\P(\neg E_n)$ converges to $1$.
\end{proof}

Finally, it remains only to use Lemma \ref{lemma:cnfxor_unsat} to obtain bounds on the $k$-CNF-XOR phase-transition. The tight upper bound on $\# \F$ from Lemma \ref{lemma:abbe}.\ref{lemma:abbe:upper} corresponds to a tight upper bound on the location of the phase-transition.
\begin{lemma}
	\label{lemma:mainresult:tight:unsat}
	
	Let $k \geq 2$, and let $\alpha_k$, $\mathcal{C}_k$, and $\phi_k(r)$ be as defined in Lemma \ref{lemma:abbe}. 
	Then for all $r \in [0, \alpha_k) \backslash \mathcal{C}_k$ and $s > \phi_k(r)$, $\FQ$ is unsatisfiable w.h.p..
\end{lemma}
\begin{proof}
	Let $B_i = 2^{\phi_k(r)+1/i}$. By Lemma \ref{lemma:abbe}.\ref{lemma:abbe:upper}, $B_i^n \geq \#\F$ w.h.p.\ for all $i \geq 1$. Furthermore, $\lim_{i \to \infty} B_i = 2^{\phi_k(r)}$ and so $s > \log_2(\lim_{i \to \infty} B_i)$. Thus $\FQ$ is unsatisfiable w.h.p.\ by Lemma \ref{lemma:cnfxor_unsat}.
\end{proof}

The weaker upper bound on $\# \F$ from Lemma \ref{lemma:numsol_bound} corresponds to a weaker upper bound on the phase-transition.
\begin{lemma}
	\label{lemma:mainresult:loose:unsat}
	Let $k \geq 2$, $s \geq 0$, and $r \geq 0$. If $s > 1 + r \log_2(1 - 2^{-k})$, then $\FQ$ is unsatisfiable w.h.p..
\end{lemma}
\begin{proof}
	Let $B_i = ((1 + 1/i) \cdot 2(1 - 2^{-k})^r)$. By Lemma \ref{lemma:numsol_bound}, $B_i^n \geq \#\F$ w.h.p.\ for all $i \geq 1$. Furthermore, $\lim_{i \to \infty} B_i = 2(1 - 2^{-k})^r$ and so $s > \log_2(\lim_{i \to \infty} B_i)$. Thus $\FQ$ is unsatisfiable w.h.p.\ by Lemma \ref{lemma:cnfxor_unsat}.
\end{proof}

\section{Extending the Phase-Transition Region}
\label{sec:conjecture}

Section \ref{sec:theoretical} proved that a phase-transition exists for $k$-CNF-XOR formulas when the $k$-clause density is small. Our empirical observations in Section \ref{sec:experiments} suggest that a phase-transition exists for higher $k$-clause densities as well. 
Moreover, we observed a phase-transition that follows a linear trade-off between $k$-clauses and XOR-clauses.
In this section, we conjecture two possible extensions to our theoretical results.

The first extension follows from Theorem \ref{thm:mainresult:tight}, which implies that $s = \phi_k(r)$ gives the location of the phase-transition for small $k$-clause densities. It is thus natural to conjecture that $\phi_k(r)$ gives the location of the $k$-CNF-XOR phase-transition for all (except perhaps countably many) $r > 0$. This would follow from a conjecture of \cite{AM14}.

The second extension follows from the experimental results in Section \ref{sec:experiments}, which suggest that the location of the phase-transition follows a linear trade-off between $k$-clauses and XOR-clauses. This leads to the following conjecture:
\begin{conjecture}[$k$-CNF-XOR Linear Phase-Transition Conjecture]
	Let $k \geq 2$. Then there exists a slope $L_k < 0$ and a constant $\alpha_k^* > 0$ such that for all $r \in [0, \alpha_k^*)$ and $s \geq 0$:
	\begin{enumerate}
		\item If $s < r L_k + 1$, then w.h.p.\ $\FQ$ is satisfiable.
		\item If $s > r L_k + 1$, then w.h.p.\ $\FQ$ is unsatisfiable.
	\end{enumerate}
\end{conjecture}

Theorem \ref{thm:mainresult:loose} bounds the possible values for $L_k$. Moreover, if the Linear $k$-CNF-XOR Phase-Transition Conjecture holds, then Theorem \ref{thm:mainresult:tight} implies that $\phi_k(r)$ is linear for all $r < \alpha_k$ and $r < \alpha_k^*$. Explicit computations of $\phi_k(r)$ (or sufficiently tight bounds) would resolve this conjecture.

Note that this conjecture does not necessarily describe the entire $k$-CNF-XOR phase-transition; a phase-transition may exist when $r > \alpha_k^*$ as well. The experimental results in Section \ref{sec:experiments} for $k = 2$ and $k = 3$ suggest that the location of the phase-transition may ``kink'' and become non-linear for large enough $k$-clause densities. We leave the full characterization of the $k$-CNF-XOR phase-transition for future work, noting that a full characterization would resolve the Satisfiability Phase-Transition Conjecture.

\section{Conclusion} \label{sec:conclusion}
We presented the first study of phase-transition phenomenon in the satisfiability of $k$-CNF-XOR random formulas. We showed that the free-entropy density $\phi_k(r)$ of $k$-CNF formulas gives the location of the phase-transition for $k$-CNF-XOR formulas when the density of the $k$-CNF clauses is small. We conjectured in the $k$-CNF-XOR Linear Phase-Transition Conjecture that this phase-transition is linear. We leave further analysis and proof of this conjecture for future work.

Pittel and Sorkin \cite{PS15} recently identified the location of the phase-transition for random $\ell$-XOR formulas, where each clause contains exactly $\ell$ literals. This suggests that a phase-transition may also exist in formulas that mix $k$-CNF clauses together with $\ell$-XOR clauses.

In this work we did not explore the runtime of {\SAT} solvers over the space of $k$-CNF-XOR formulas. Historically, other phase-transition phenomena have been closely connected empirically to solver runtime. Developing this connection in the case of $k$-CNF-XOR formulas is an exciting direction for future research and may lead to practical improvements to hashing-based sampling and counting algorithms.

\subsection*{Acknowledgments}
 The authors would like to thank Dimitris Achlioptas for helpful discussions in the early stages of this project.
 
Work supported in part by NSF grants CCF-1319459 and IIS-1527668, by NSF Expeditions in Computing project ``ExCAPE: Expeditions in Computer Augmented Program Engineering", by BSF grant 9800096,  by the Ken Kennedy Institute Computer Science \& Engineering
 Enhancement Fellowship funded by the Rice Oil \& Gas HPC Conference, and Data Analysis and 
 Visualization Cyberinfrastructure funded by NSF under grant 
 OCI-0959097. Part of this work was done during a visit to the Israeli Institute for Advanced Studies.

\end{document}